\documentclass[10pt,journal]{IEEEtran}
\usepackage{cite}
\usepackage{graphicx}
\usepackage{psfrag}
\usepackage{subfigure}
\usepackage{url}
\usepackage[cmex10]{amsmath}
\usepackage{array}
\usepackage{graphics}
\usepackage{epsfig}
\usepackage{amsbsy}
\usepackage{amssymb}
\usepackage{amsthm}

\usepackage{color}
\newcounter{MYtempeqncnt}
\newtheorem{proposition}{Proposition} 
\usepackage[normalem]{ulem}
\IEEEoverridecommandlockouts
\begin{document}
\title{\LARGE  Sensing-Throughput Tradeoff in Cognitive Radio\\ With Random Arrivals and Departures of Multiple Primary Users}

\author{Hrusikesha~Pradhan,
        Sanket~S.~Kalamkar,~\IEEEmembership{Student~Member,~IEEE,}
        and~Adrish~Banerjee,~\IEEEmembership{Senior~Member,~IEEE}   
     
\thanks{The authors are with the Department of Electrical Engineering, IIT Kanpur, 208016, India. (e-mail:  $\lbrace$kalamkar, adrish$\rbrace$@iitk.ac.in).}
}
\maketitle  

\begin{abstract}
This letter analyzes the sensing-throughput tradeoff for a secondary user (SU) under random arrivals and departures of multiple primary users (PUs). We first study the case where PUs change their status only during SU's sensing period. We then generalize to a case where PUs change status anytime during SU frame, and compare the latter case with the former in terms of the optimal sensing time and SU throughput. We also investigate the effects of PU traffic parameters and the number of PUs on the sensing-throughput tradeoff for SU. Results show that, though the increase in the number of PUs reduces the optimal sensing time for SU, the opportunity to find a vacant PU channel reduces simultaneously, in turn, reducing SU throughput. Finally, we validate the analysis by Monte Carlo simulations.
\end{abstract}
\section{Introduction}
\PARstart{S}{pectrum} sensing~\cite{survey} plays a key role in cognitive radio. The energy detection (ED) is one of the most widely used spectrum sensing techniques as it needs no knowledge about the primary user (PU), and has low complexity~\cite{survey}. The frame of a secondary user (SU) consists of a sensing period in which SU performs spectrum sensing to detect whether PU channel is busy or idle; if found idle, SU may transmit over it in the transmission period of the frame. Longer sensing period leads to better sensing performance, but at the cost of decrease in the transmission period for SU, thereby reducing its throughput. This results in the sensing-throughput tradeoff for SU~\cite{liang}. 

Most previous works assume that PU maintains constant occupancy throughout the SU frame. However, PU may arrive or depart anytime during the frame due to high PU traffic or long SU frame duration~\cite{chen1}. In~\cite{chen3, wu, tran,chen5}, authors study the sensing performance of ED against the random arrival and departure of one PU, while~\cite{chen1} and \cite{lu} examine the effect of a single PU traffic on SU's sensing performance as well as throughput. In~\cite{chen2}, authors consider the effect of multiple PUs traffic on the sensing performance of SU when PUs change their status only during the sensing period. However, authors in~\cite{chen2} ignore the effect of PU traffic on SU throughput. It is possible that PUs may change their status anytime during SU frame, including the transmission period. In this case, even if SU finds the channel idle at the end of its sensing period, PUs may arrive in the transmission period interfering SU's transmission. Therefore, it is important to study the effects of random arrivals and departures of multiple PUs on SU's sensing performance as well as throughput. 

In this letter, we first derive closed-form expressions of probabilities of detection and false alarm, and SU throughput for cases when multiple PUs change their status only during the sensing period and anytime during SU frame. Then, we evaluate the joint sensing-throughput performance of SU in terms of sensing-throughput tradeoff, and compare performances of both aforementioned cases. To the best of our knowledge, this letter is the first work to study the effect of multiple PUs traffic on the joint sensing-throughput performance of SU. Finally, we show the effects of  number of PUs and their traffic parameters such as arrival and departure rates on the sensing-throughput tradeoff.   

\begin{figure*}[!t]
\normalsize
\setcounter{MYtempeqncnt}{\value{equation}}
\setcounter{equation}{3}
\begin{equation}
\label{eq:p_h1imk1}
P\left( { \mathcal{H} }_{ 1,i}, m, k \right) = M{ p }_{ Bm }{ \left( 1-{ F }_{ \alpha  }\left( T_{\mathrm{s}} \right)  \right)  }^{ k }\left(\sum_{a_1 = 0}^{L-1}\dotsc\sum_{a_{i-k} = 0}^{L-1}\sum_{d_1 = 0}^{L-1}\dotsc\sum_{d_{m-k} = 0}^{L-1}\prod _{ j = 1 }^{i -k }{ { p }_{ \beta  }\left( { a }_{ j } \right)  }   
\prod _{ l=1 }^{ m-k }{ { p }_{ \alpha  }\left( { d }_{ l } \right)  }\right) { \left( 1-{ F }_{ \beta  }\left( T_{\mathrm{s}} \right)  \right)  }^{ N-\left(i -k \right) -m }.
\end{equation}
\setcounter{equation}{\value{MYtempeqncnt}}
\hrulefill
\end{figure*}
\section{System Model}
Consider a scenario where multiple PUs occupy the same frequency band and a SU attempts to use PU band opportunistically, as in~\cite{chen2}. The ED is used for spectrum sensing, whose output is given as $\Lambda =\sum _{ j=1 }^{ L }{ { { Y }_{ j } }^{ 2 } }$, where ${ Y }_{ j }={ W }_{ j }$ denotes $j$th received sensing sample by SU when the channel is idle, while ${ Y }_{ j }=\sum_{z=1}^{u}{ X }_{ j, z }+{ W }_{ j }$ denotes $j$th received sensing sample when the channel is busy. ${ W }_{ j }$ denotes additive white Gaussian noise (AWGN) with mean zero and variance $\sigma^2$. The term $\sum_{z=1}^{u}{ X }_{ j, z }$ represents signals from $u$ PUs at $j$th sample, with $\mathcal{P}$ being the average power of each PU. Let $\gamma_{\mathrm{p}} = \mathcal{P}/\sigma^2$ be the signal-to-noise ratio (SNR) of one PU at SU. Then, the received SNR at SU when $u$ PUs are present is $u\gamma_{\mathrm{p}}$\footnote{For ease of representation and without compromising the insight into analysis, we assume that each PU has the same SNR $\gamma_{\mathrm{p}}$ as in \cite{chen2}. When each PU has different SNR, the received SNR at SU when $u$ PUs are present becomes $\sum_{z = 1}^{u}\gamma_{\mathrm{p}z}$, where $\gamma_{\mathrm{p}z}$ is the received SNR at SU due to $z$th PU.} \cite{chen2}. The PU remains either in busy ($\alpha$) or idle ($\beta$) state. The holding times of busy and idle periods are random in nature and exponentially distributed~\cite{chen1,chen3,lu} with means $\theta_{\alpha}$ and $\theta_{\beta}$; while ${ F }_{ \alpha  }\left( x \right)$ and ${ F }_{ \beta  }\left( x \right)$ denote their respective cumulative distributions. The probabilities of a PU being busy and idle are ${ p }_{ \mathrm{b} }= \theta_{\alpha}/(\theta_{\alpha} +\theta_{\beta})$ and ${ p }_{ \mathrm{i} }=\theta_{\beta}/(\theta_{\alpha} +\theta_{\beta})$, respectively. The probability mass function (PMF) when a PU changes its status from idle to busy after $j$th sample is\cite{ma}
\begin{equation}
\label{eq:idletobusy}
{ p }_{ \beta  }\left( j \right) ={ F }_{ \beta  }\left( \left( j+1 \right) t_{ \mathrm{s} } \right) -{ F }_{ \beta  }\left( j{ t }_{ \mathrm{s} } \right),
\end{equation}
where $t_\mathrm{s}$ is the sampling interval. Similarly, PMF when a PU changes its status from busy to idle after $j$th sample is
\begin{equation}
\label{eq:busytoidle}
{ p }_{ \alpha  }\left( j \right) ={ F }_{ \alpha  }\left( \left( j+1 \right) t_{ \mathrm{s} } \right) -{ F }_{ \alpha  }\left( j{ t }_{ \mathrm{s} } \right).
\end{equation}

\textit{Notation}: We define $\sum_{t = a}^{b}(\cdot) = 0$, when $b < a$.

\section{Problem formulation and Derivations}
Let $N$ denote the number of PUs. Thus, at the beginning of the sensing period, {\em i.e.}, at the beginning of the SU frame, the channel can have $N+1$ possible states with probabilities,
\begin{equation}
\label{eq:pbmusers}
{ p }_{ Bm }=\left( \begin{matrix} N \\ m \end{matrix} \right) { p }_{ \mathrm{b} }^{ m }{ p }_{ \mathrm{i} }^{ N-m } \hspace*{2mm} \text{and} \hspace*{2mm}{ p }_{ I }=\left( \begin{matrix} N \\ N \end{matrix} \right) { p }_{ \mathrm{i} }^{ N },
\end{equation}
where $p_{Bm}$ is the probability of $m$ PUs occupying the channel at the start of the SU frame ($m = 1, \dotsc, N$), and $p_{I}$ is the probability that all PUs are idle. Each PU changes its status at most once \cite{chen1,chen2, lu} during the entire frame period. Thus, at most $N$ PU status changes may occur in a SU frame.
\subsection{Status Change Only During Sensing}
\label{sec:1}
In this section, we consider a scenario where PUs change their status only during the sensing period.  Let $\left( { \mathcal{H} }_{ 1,i},m,k \right)$ be the hypothesis that $m$ PUs are busy at the start of the sensing period, $i$ PUs are busy at the end of sensing period,  and $k$ out of $m$ PUs remain busy throughout the sensing period. This means $m-k$ PUs leave the channel during the sensing period. Thus, to have $i$ PUs present at the end of the sensing period, we need $i-k$ PUs to arrive during the sensing period. Then, $N-(i-k)-m$ represents the number of PUs that remain idle throughout the sensing period. The probabilities that $k$ PUs remain present and $N-(i-k)-m$ PUs remain absent throughout the sensing period are $\left( 1-{ F }_{ \alpha  }\left( T_{\mathrm{s}} \right)\right)^k$ and $\left( 1-{ F }_{ \beta  }\left( T_{\mathrm{s}} \right)\right)^{N-(i-k)-m}$, respectively, where $T_{\mathrm{s}}$ is the sensing period. Let $a_j$ and $d_l$ denote the sample after which $j$th PU arrives and $l$th PU departs, respectively, and $a_j$, $d_l$  $ = 0, 1, \dotsc, L-1$, where $L$ is the number of sensing samples. Considering all possible combinations under the hypothesis $\left( { \mathcal{H} }_{ 1,i},m,k \right)$, we can write the probability of the hypothesis $\left( { \mathcal{H} }_{ 1,i},m,k \right)$ as \eqref{eq:p_h1imk1}, where $M = \left( \begin{matrix} N \\ N-k \end{matrix} \right) $ when all PUs are busy at the start of the sensing period, because $N-k$ out of $N$ PUs leave the channel and $M$ represents its all possible combinations. When all PUs are idle at the start of the sensing period, $M = \left( \begin{matrix} N \\ i \end{matrix} \right) $, {\em i.e.}, $M$ here represents all possible combinations of $i$ out of $N$ PUs arriving in the channel. In all other cases, $M = 1$ as ${ p }_{ Bm }$ captures all possible combinations. \addtocounter{equation}{1}

\begin{proposition}
The value of $k$ ranges from $\max\left( 0, m+i -N \right) $ to $\min\left( m, i  \right) $.
\label{lemma1}
\end{proposition}
\begin{proof}
To find the maximum value that $k$ may take, we have to consider following two cases:
\begin{itemize}
\item[I)]$i > m$: Here, the maximum value that $k$ may take is $m$.
\item[II)]$i \leq m$: Here, the maximum value that $k$ may take is $i$.
\end{itemize}
Thus, combining both cases, we can infer that the maximum value that $k$ may take is $\min(m,  i)$. 

The minimum value of $k$ corresponds to the maximum number of arrivals of PUs, which can be represented by the equality $i - k =  N - m$. This gives the minimum value of $k = m + i -N$. However, the minimum value of $k$ should be non-negative. Thus, it can be given as $\max(0,  m + i -N)$.
\end{proof}
Taking an example, assume number of PUs to be 10, {\em i.e.}, $N = 10$. Suppose 8 PUs are present at the beginning of sensing period, {\em i.e.}, $m = 8$, and 3 PUs are present at the end of sensing period, {\em i.e.}, $i$ = 3. Then, at most 3 out of 8 PUs may stay busy during the whole sensing period, {\em i.e.}, $\min(m, i)$. Since maximum 2 new PUs can arrive, there has to be minimum 1 out of 8 PUs occupying the channel throughout the sensing period, {\em i.e.}, $\max(0, m+i-N)$, so that when 7 PUs depart, 3 PUs remain present at the end of the sensing period.

Now, let $\left( { \mathcal{H} }_{ 1,i},m \right)$ denote the hypothesis that $i$ PUs are busy at the end of the sensing period and $m$ users are busy at the start of the sensing period. Summing over all possible values of $k$, we can write the probability of $\left( { \mathcal{H} }_{ 1,i},m \right)$ as
 \begin{equation}
\label{eq:p_H1im1}
P\left( { \mathcal{H} }_{ 1,i},m \right) =\sum _{ k=\max\left( 0, m+i -N \right)  }^{ \min\left( m, i  \right)  }{ P\left( { \mathcal{H} }_{ 1,i }, m, k \right)  }.
\end{equation} 
The probability of hypothesis $\mathcal{H}_{1,i}$ that $i$ PUs occupying the channel at the end of the sensing period, given $m$ = 0, 1, $\dotsc, N$ PUs are present at the start of the sensing period is 
\begin{equation}
P\left( { \mathcal{H} }_{ 1,i } \right) =\sum _{ m=0 }^{ N }{ P\left( { \mathcal{H} }_{ 1,i },m \right)  }.
\label{eq:H_1ii}
\end{equation}
Let ${{\mathcal{H}}_1}$ denote the hypothesis that the channel is occupied at the end of the sensing period, {\em i.e.}, $1 \leq i \leq N$. Then, the probability of hypothesis ${{\mathcal{H}}_1}$ can be given as 
\begin{equation}
\label{eq:p_H1i}
P\left( { \mathcal{H} }_{ 1 } \right) = \sum_{i = 1}^{N} P\left( { \mathcal{H} }_{ 1,i } \right). 
\end{equation}
Let ${{\mathcal{H}}_0}$ denote the hypothesis that the channel is idle at the end of the sensing period, {\em i.e.}, $i = 0$. The probability of hypothesis  $\mathcal{H}_0$, {\em i.e.}, $P\left( { \mathcal{H} }_{ 0 } \right)$, can also be obtained from \eqref{eq:H_1ii} by substituting $i = 0$. Then, the probability of $\mathcal{H}_0$ can be written as
\begin{equation}
\label{eq:p_H0i}
P\left( { \mathcal{H} }_{ 0 } \right) =\sum _{ m=0 }^{ N }{ P\left( { \mathcal{H} }_{ 1,i },m \right) \quad \text{with}\quad i = 0 }. 
\end{equation}

\begin{figure*}
\normalsize
\setcounter{MYtempeqncnt}{\value{equation}}
\setcounter{equation}{9}
\begin{equation}
P\left( { \mathcal{H} }_{ 1 }|{ \mathcal{H} }_{ 1,i },m,k \right) =\frac { 1 }{ 2 } \mathrm{erfc}\left( \frac { \eta -L-\left( kL{ \gamma  }_{ \mathrm{p} }+\sum _{ j=1 }^{ i -k }{ \left( L-{ a }_{ j } \right) { \gamma  }_{ \mathrm{p} } } +\sum _{ l=1 }^{ m-k }{ { d }_{ l }{ \gamma  }_{ \mathrm{p} } }  \right)  }{ 2\sqrt { 2 } \sqrt { \frac { L }{ 2 } +kL{ \gamma  }_{ \mathrm{p} }+\sum _{ j=1 }^{ i -k }{ \left( L-{ a }_{ j } \right) { \gamma  }_{ \mathrm{p} } } +\sum _{ l=1 }^{ m-k }{ { d }_{ l }{ \gamma  }_{ \mathrm{p} } }  }  }\right).
\label{eq:pd_imk}
\end{equation}
\setcounter{equation}{\value{MYtempeqncnt}}
\hrulefill 
\end{figure*}
\begin{figure*}
\normalsize
\setcounter{MYtempeqncnt}{\value{equation}}
\setcounter{equation}{17}
\begin{equation}
\begin{split}
P\left( { \mathcal{H} }_{ 1,i },m,k \right) =M{ p }_{ Bm }\Bigg(\sum_{g_{1} = L}^{S}\dotsc\sum_{g_{k} = L}^{S}\sum_{a_1 = 0}^{L-1}\dotsc\sum_{a_{i-k} = 0}^{L-1}\sum_{d_1 = 0}^{L-1}\dotsc\sum_{d_{m-k} = 0}^{L-1}\sum_{c_{1}=L}^{S}\dotsc\sum_{c_{N-(i-k)-m}=L}^{S}\prod _{ \phi =1 }^{ k }{ { p }_{ \alpha  }\left( { g }_{ \phi } \right)  }\prod _{ j=1 }^{ i -k }{ { p }_{ \beta  }\left( { a }_{ j } \right)  } \\
\times \prod _{ l=1 }^{ m-k }{ { p }_{ \alpha  }\left( { d }_{ l } \right)  }\prod _{ \varphi=1 }^{ N-\left( i -k \right) -m }{ { p }_{ \beta  }\left( { c }_{ \varphi } \right)}\Bigg).
\end{split}
\label{eq:p_h1imk}
\end{equation}
\setcounter{equation}{\value{MYtempeqncnt}}
\hrulefill
\end{figure*}
\begin{figure*}
\normalsize
\setcounter{MYtempeqncnt}{\value{equation}}
\setcounter{equation}{18}
\begin{equation}
C\left( { \mathcal{H} }_{ 1,i },m,k \right) =\log _{ 2 }{ \left( 1+\frac { { \gamma  }_{ \mathrm{s} } }{ 1+\left(\sum _{ \phi=1 }^{ k }{ \frac { { g }_{ \phi }-L }{ S-L } { \gamma  }_{ \mathrm{p} }+\left( i -k \right) { \gamma  }_{ \mathrm{p} }+\sum _{ \varphi=1 }^{ N-\left( i -k \right) -m }{ \frac { S-{ c }_{ \varphi } }{ S-L }  } { \gamma  }_{ \mathrm{p} } } \right) }  \right)  }.
\label{eq:cap}
\end{equation} 
\setcounter{equation}{\value{MYtempeqncnt}}
\hrulefill
\end{figure*}

For energy detection, when there is one PU ($N = 1$), the generalized expression of the probability of detection and the probability of false alarm for AWGN channel under any hypothesis can be given as\cite[(8) and (10)]{chen1}\footnote{Basically, \eqref{eq:det} denotes the probability that SU senses the presence of PU at the end of the sensing period. When PU is present, the term $n\gamma_{\mathrm{p}}$ is non-zero, and \eqref{eq:det} represents the probability of detection. When PU is absent, $\gamma_{\mathrm{p}}$ in \eqref{eq:det} becomes zero, making $n\gamma_{\mathrm{p}} = 0$. In this case, \eqref{eq:det} represents the probability of false alarm.}
\begin{equation}
\frac { 1 }{ 2 } \mathrm{erfc}\left( \frac { \eta -L-n{ \gamma  }_\mathrm{p} }{ 2\sqrt { 2 } \sqrt { \frac { L }{ 2 } +n{ \gamma  }_\mathrm{p} }  }  \right),
\label{eq:det}
\end{equation}  
where $\eta$ is the detection threshold and $n$ represents the number of samples for which PU occupies the channel in the sensing period. Thus, the term $n\gamma_{\mathrm{p}}$ corresponds to the energy received from PU's transmission. $\mathrm{erfc}(\cdot)$ is the complementary error function given by $\mathrm{erfc}(x) = (2/\pi)\int_{x}^{\infty}\exp(-t^2)\mathrm{d}t$. Then, under the hypothesis ($\mathcal{H}_{1,i}, m, k$) for multiple PUs, the energy received from $k$ PUs that remain busy throughout the sensing period of $L$ samples will correspond to $kL\gamma_{\mathrm{p}}$; the energy received from $i-k$ arriving PUs during the sensing period will correspond to $\sum_{j=1}^{i-k}{\left(L-{a}_{j} \right){\gamma}_{\mathrm{p}}}$; and the energy received from $m-k$ departing PUs during the sensing period will correspond to $\sum _{ l=1 }^{ m-k }{ { d }_{ l }{ \gamma  }_{ \mathrm{p} } } $. Then, in the case of multiple PUs, we can generalize \eqref{eq:det} further to obtain the expression for the conditional probability of detection $P\left( { \mathcal{H} }_{ 1 }|{ \mathcal{H} }_{ 1,i },m,k \right)$ given by \eqref{eq:pd_imk}, under the hypothesis ($\mathcal{H}_{1,i}, m, k$) with $i \neq 0$ PUs busy at the end of the sensing period. \addtocounter{equation}{1}

The unconditional probability of detection can be derived by averaging the conditional probability of detection over probabilities of respective hypotheses as 
\begin{align}
P_{\mathrm{ d} }&=\frac { 1 }{ P\left( { \mathcal{H} }_{ 1 } \right)  } \sum _{ i=1 }^{ N } \sum _{ m=0 }^{ N } \sum _{ k=\max\left( 0,m+i -N \right)  }^{ \min\left( m, i  \right)  }P\left( { \mathcal{H} }_{ 1,i },m,k \right)     \nonumber \\ 
&\times P\left( { \mathcal{H} }_{ 1 }|{ \mathcal{H} }_{ 1,i },m,k \right).
\label{eq:pd}
\end{align}
The probability of false alarm is the probability of falsely detecting the presence of at least one PU at the end of the sensing period when, in actual, no PU is present, {\em i.e.}, $i = 0$. Then, the unconditional probability of false alarm is same as \eqref{eq:pd}, but with $i = 0$, and can be given as
\begin{align}
P_{ \mathrm{f} }&=\frac { 1 }{ P\left( { \mathcal{H} }_{ 0 } \right)  }\Bigg(\sum _{ m=0 }^{ N }  {\sum _{ k=\max\left( 0,m+i -N \right)  }^{ \min\left( m, i  \right)  }}P\left( { \mathcal{H} }_{ 1,i },m,k \right) \nonumber \\
&  \times { P\left( { \mathcal{H} }_{ 1 }|{ \mathcal{H} }_{ 1,i },m,k \right)  }  
  \Bigg), 
 \quad \text{for}\hspace*{2mm} i = 0.
 \label{eq:pf1}
\end{align}
After sensing, if SU finds the channel idle, it may begin the transmission. As each PU changes its state only during the sensing period and maintains the same state in the transmission period as it was at the end of the sensing period, the channel capacity under the hypothesis $\mathcal{H}_{1,i}$ when $i$ PUs are occupying the channel at the end of the sensing period can be given by
\begin{equation}
C\left( { \mathcal{H} }_{ 1,i } \right) =\log _{ 2 }{ \left( 1+\frac { { \gamma  }_{ \mathrm{s} } }{ 1+i { \gamma  }_{ \mathrm{p} } }  \right)  },
\label{eq:cap11}
\end{equation}
where $\gamma_{\mathrm{s}}$ is SNR of the secondary transmission. The capacity under the hypothesis $\mathcal{H}_{0}$ when no PU is present at the end of the sensing period ($i = 0$) is given by
\begin{equation}
C\left( { \mathcal{H} }_{ 0 } \right) =\log _{ 2 }{ \left( 1+{ \gamma  }_{ \mathrm{s} } \right)  }. 
\end{equation} 

Under the hypothesis $ { \mathcal{H} }_{1}$, when at least one PU is present at the end of the sensing period, {\em i.e.}, $i \neq 0$, we can obtain the average SU throughput by averaging $C\left( { \mathcal{H} }_{ 1,i } \right)$ in \eqref{eq:cap11} over the probability of occurrence, and is given as
\begin{equation}
R({\mathcal{H}_{1}}) = \left(1-{P}_{\mathrm{d}}\right) \frac{{T}_{\mathrm{f}}-{T}_{\mathrm{s}}}{{T}_{\mathrm{f}}}\sum_{i = 1}^{N} P\left( {\mathcal{H}}_{1,i}\right) C\left({\mathcal{H}}_{1,i}\right),
\label{eq:cap111}
\end{equation} 
where $T_{\mathrm{f}}$ is the frame period. Note that SU achieves $R({\mathcal{H}_{1}})$ only in the case of miss-detection, {\em i.e.}, when it fails to detect the presence of PU at the end of the sensing period. Similarly, under the hypothesis $ { \mathcal{H} }_{0}$, the average SU throughput can be given as
\begin{equation}
R\left({\mathcal{H}}_{0}\right) =\left(1-{P}_{\mathrm{f}}\right) \frac {{T}_{ \mathrm{f}}-{T}_{\mathrm{s}}}{{T}_{\mathrm{f}}} P\left({\mathcal{H}}_{0}\right) C\left({\mathcal{H}}_{0}\right). 
\end{equation}
Then the average achievable SU throughput is given by
\begin{equation}
R={ R\left({\mathcal{H}}_{1}\right)} + R\left({\mathcal{H}}_{0}\right). 
\end{equation}

\subsection{Status Change Anytime During Frame}
We now consider a generalized case, where multiple PUs change their status anytime during the frame. In this case, we can write the probability $P\left( { \mathcal{H} }_{ 1,i},m,k \right)$ that $i$ PUs are busy at the end of the sensing period when $m$ PUs are present at the start of the sensing period, {\em i.e.}, at the start of the frame, and $k$ out of $m$ PUs remain busy throughout the sensing period, by \eqref{eq:p_h1imk}. The equation \eqref{eq:p_h1imk} is derived following the steps used in deriving \eqref{eq:p_h1imk1}. In \eqref{eq:p_h1imk}, $g_{\phi}$, $c_{\varphi}$ $= L, L+1, \dotsc, S$, where $S$ is the number of samples in a frame. All other notations have same meanings as they have in Section\,\ref{sec:1}. The terms ${ p }_{ \alpha  }\left( { g }_{ \phi } \right) $ and ${ p }_{ \beta  }\left( { c }_{ \varphi } \right) $ denote probabilities of transition from busy to idle after ${ g }_{ \phi } $th sample and idle to busy after ${ c}_{ \varphi }$th sample, respectively, during the transmission period. $g_{\phi} = S$ and $c_{\varphi} = S$ denote that PU remains busy and idle, respectively, throughout the transmission period. Then,  ${ p }_{ \alpha  }\left( { g }_{ \phi } \right) $ and ${ p }_{ \beta  }\left( { c }_{ \varphi } \right) $ become $1-{ F }_{ \alpha  }\left( { T }_{ \mathrm{f}} \right) $  and  $1-{ F }_{ \beta  }\left( { T }_{ \mathrm{f} } \right) $, respectively.

We can obtain the probability $P\left( { \mathcal{H} }_{ 1,i},m \right)$ by substituting the value of $P\left( { \mathcal{H} }_{ 1,i },m,k \right) $ from \eqref{eq:p_h1imk} in \eqref{eq:p_H1im1}. 
Then, $P\left( { \mathcal{H} }_{ 1,i } \right)$, $P\left( { \mathcal{H} }_{ 1 } \right)$, and $P\left( { \mathcal{H} }_{ 0} \right)$ can be found using \eqref{eq:H_1ii}, \eqref{eq:p_H1i}, and \eqref{eq:p_H0i}, respectively, as given in Section\,\ref{sec:1}. The probability of detection $P_{\mathrm{d}}$ and the probability of false alarm $P_{\mathrm{f}}$ can be obtained by substituting \eqref{eq:p_h1imk} in \eqref{eq:pd} and \eqref{eq:pf1}, respectively.

The channel capacity under the hypothesis (${ \mathcal{H} }_{ 1,i },m,k$) can be given by \eqref{eq:cap}, where the term $\sum _{ \phi=1 }^{ k } \frac { { g }_{ \phi }-L }{ S-L } { \gamma  }_{ \mathrm{p} }$ corresponds to the departures of PUs in the transmission period who were present throughout the sensing period;  the term $\left( i -k \right) { \gamma  }_{ \mathrm{p} }$ corresponds to PUs who were present at the end of the sensing period as well as remain present throughout the transmission period; the term $\sum_{ \varphi=1 }^{ N-\left( i -k \right) -m }{ \frac { S-{ c }_{ \varphi } }{ S-L }  } { \gamma  }_{ \mathrm{p} } $ corresponds to the arrivals of PUs in the transmission period. \addtocounter{equation}{2}Accordingly, the average SU throughput under the hypothesis $ { \mathcal{H} }_{ 1,i }$ can be given by
\begin{equation}
R\left( { \mathcal{H} }_{ 1,i } \right) = \sum _{ m = 0 }^{ N }{ \sum _{ k = \max(0, m+i -N )}^{ \min\left( m,i  \right)  }{ R\left( { \mathcal{H} }_{ 1,i },m,k \right)  }  },
\end{equation}
where $R\left( { \mathcal{H} }_{ 1,i },m,k \right) $ is given as
\begin{align}
R\left( { \mathcal{H} }_{ 1, i }, m, k \right) &=\left( 1-{ P }_{ \mathrm{d} } \right) \frac { { T }_{\mathrm{f} }-{ T }_{\mathrm{s} } }{ { T }_{ \mathrm{f} } } P\left( { \mathcal{H} }_{ 1, i }, m, k \right) \nonumber\\
&\times C\left( { \mathcal{H} }_{ 1, i }, m, k \right).
\end{align}
Then the achievable throughput when at least one PU is busy at the end of the sensing period is given by
\begin{equation}
{R}(\mathcal{H}_{1})=\sum _{ i = 1 }^{ N }{ R\left( { \mathcal{H} }_{ 1, i } \right)  }.
\end{equation}
The achievable throughput when there is no PU at the end of the sensing period is
\begin{equation}
{ R }(\mathcal{H}_{0}) = \left( 1-{P}_{\mathrm{f}} \right) \frac { { T }_{ \mathrm{f} }-{ T }_{ \mathrm{s} } }{ { T }_{ \mathrm{f} } } P\left( { \mathcal{H} }_{ 1,i }, m, k \right) C\left( { \mathcal{H} }_{ 1, i }, m, k \right),
\end{equation}
 with $i = 0$. Thus, the average achievable SU throughput becomes
\begin{equation}
R=R(\mathcal{H}_1)+R(\mathcal{H}_0).
\end{equation}

\section{Results and Discussions}
In this section, we evaluate the effect of multiple PUs with random arrivals and departures, assuming \textbf{Case I}:
the status change only during sensing, and \textbf{Case II}: the status change anytime during frame, on the sensing-throughput tradeoff for SU. The frame period is $T_{\mathrm{f}}$ = 30$\text{ms}$ and the sampling interval is $t_\mathrm{s}$ = 100$\mathrm{\mu s}$. The SNR of each PU is $\gamma_{\mathrm{p}}$ = $-5\text{dB}$. The SNR of SU is $\gamma_{\mathrm{s}}$ = $10\text{dB}$. The noise variance is $\sigma^2 = 1$. Using Neyman-Pearson rule\cite{kay}, the detection threshold in ED is found by minimizing the probability of false alarm for the target probability of detection $P_{\mathrm{d}}$ = 0.9. We verify the analysis by simulations. 
  
Fig.\,\ref{fig:edge-b1} compares the average achievable SU throughput for both aforementioned cases of status changes of PUs. We can see from Fig.~\ref{fig:edge-b1} that the SU throughput in Case II is lower than that in Case I. This is because, in Case II, even if SU correctly detects the channel is idle at the end of the sensing period and starts transmitting during transmission period, PUs may arrive in the channel during transmission period causing interference to SU, in turn, reducing its throughput. However, the interference caused due to arrivals of PUs does not occur in Case I. For Case II, though PUs may also depart in the transmission period, this will happen only in the event of miss-detection by SU, whose probability of occurrence is very less compared to arrivals of PUs in the transmission period, due to high target probability of detection. Figs.~\ref{fig:edge-b1} and \ref{fig:edge-a1} also show that, with the increase in the number of PUs, the optimal sensing time reduces as PU SNR received at SU increases, in turn, helping SU's throughput to increase. However, at the same time, increase in the number of PUs reduces the chances of SU finding the channel idle, reducing SU throughput. The latter effect is more pronounced as seen from Figs.~\ref{fig:edge-b1} and \ref{fig:edge-a1}. Also, for Case II, increase in the number of PUs increases the probability of arrivals of PUs in the transmission period of SU, in turn, increasing the probability of interference. This further reduces SU throughput.

As shown in Fig.\,\ref{fig:edge-a1}, decrease in average holding times for busy ($\theta_{\alpha}$) and idle ($\theta_{\beta}$) states increases PU traffic, which leads to decrease in SU throughput.
   \begin{figure}
  \centering
    \subfigure[]{\label{fig:edge-b1}\includegraphics[scale=0.38]{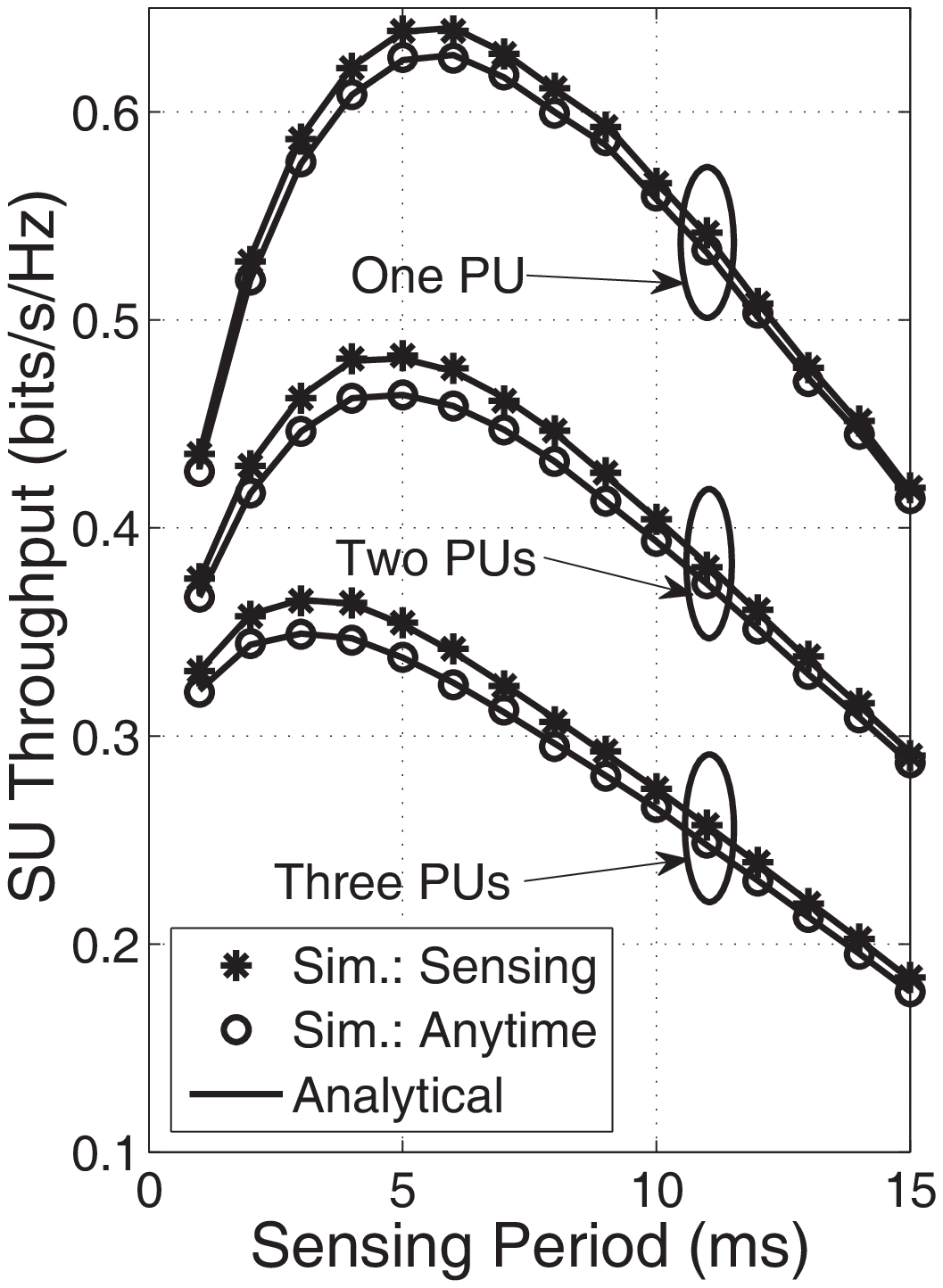}} 
        \subfigure[]{\label{fig:edge-a1}\includegraphics[scale=0.38]{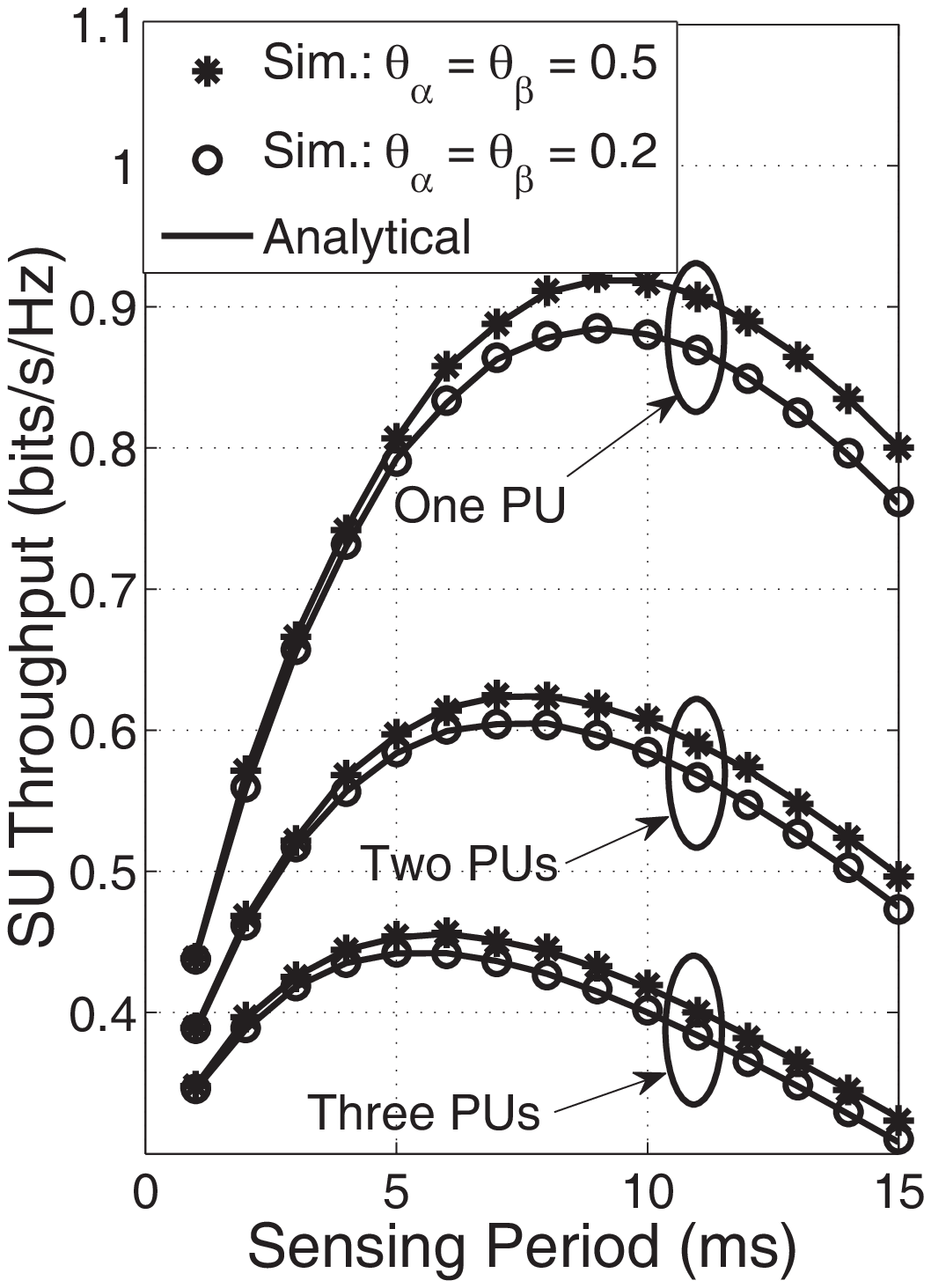}}
      \caption{(a) Sensing-throughput tradeoff assuming status change only during sensing and status change anytime during frame, $\theta_{\alpha}$ = $\theta_{\beta}$ = $0.02$. (b)  Sensing-throughput tradeoff for different PU traffic parameters with status change anytime during frame.}
      \label{fig:ANDrule}
         \end{figure}
\bibliographystyle{ieeetr}
\bibliography{paper}

\end{document}